\documentclass[conference]{IEEEtran}
\usepackage{ifpdf}
\usepackage{cite}
\ifCLASSINFOpdf
   \usepackage[pdftex]{graphicx}
   \graphicspath{{../pdf/}{../jpeg/}}
   \DeclareGraphicsExtensions{.pdf,.jpeg,.png}
\else
   \usepackage[dvips]{graphicx}
   \graphicspath{{../eps/}}
   \DeclareGraphicsExtensions{.eps}
\fi
\usepackage{epsf,psfrag,epsfig}

\usepackage[cmex10]{amsmath}
\usepackage{mathrsfs, amsthm,amsfonts, latexsym, amssymb}
\usepackage{algorithmic,algorithm}
\usepackage{hhline}\usepackage[english]{babel}
\usepackage[mathscr]{eucal}
\usepackage{graphicx,color}
\usepackage[font=small]{caption}
\usepackage[font=footnotesize]{subcaption}
\usepackage{multicol,multirow}
\usepackage[breaklinks]{hyperref}
\usepackage[hyphenbreaks]{breakurl}
\newtheorem{lem}{Lemma}
\newtheorem{thm}{Theorem}
\IEEEoverridecommandlockouts
\begin{document}
\title{Stochastic Interchange Scheduling in the Real-Time Electricity Market}

\author{Yuting Ji, Tongxin Zheng, {\em Senior Member, IEEE}, and  Lang Tong, {\em Fellow, IEEE}\thanks {\scriptsize This work is supported in part by the DoE CERTS program and the National Science Foundation under Grant CNS-1135844. Part of this work appeared in \cite{JiTong15PESGM}.

Y. Ji and L. Tong are with the School of Electrical and Computer Engineering, Cornell University, Ithaca, NY 14853, USA (e-mail: yj246@cornell.edu; ltong@ece.cornell.edu).

T. Zheng is with ISO New England Inc., Holyoke, MA 01040, USA (e-mail: tzheng@iso-ne.com).}}
\maketitle

\begin{abstract}
The problem of multi-area interchange scheduling in the presence of stochastic generation and load is considered.   A new interchange scheduling technique based on a two-stage stochastic minimization of overall expected operating cost is proposed.  Because directly solving the stochastic optimization is intractable, an equivalent problem that maximizes the expected social welfare is formulated. The proposed technique leverages the operator's capability of forecasting locational marginal prices (LMPs) and obtains the optimal interchange schedule without iterations among operators.
\end{abstract}


\IEEEpeerreviewmaketitle
\section{Introduction}\label{sec:intro}
Since the restructuring of the electric power industry, independent system operators (ISOs) and regional transmission organizations (RTOs) have faced the seams issue characterized by  the inefficient transfer of power between neighboring regions.  Such inefficiency is caused  by incompatible market designs of independently controlled operating regions, inconsistencies of their scheduling protocols, and their different pricing models. The economic loss due to seams for the New York and New England customers is estimated at the level of \$784 million annually \cite{IRIS}.

There has been recent effort in addressing the seams issue by optimizing interchange flows across different regions.  In particular, a new interchange scheduling technique, referred to as Tie Optimization (TO),  is proposed in \cite{IRIS} to minimize the overall operating cost.  The Federal Energy Regulatory Commission (FERC)
has recently approved the Coordinated Transaction Scheduling (CTS) that allows market participants' participation in TO.  Implementations of various versions of CTS are being carried out by several system operators in the US \cite{ferc_app12}\cite{ferc_app14}.

One of the main challenges in eliminating seams is the inherent delay between the interchange scheduling and the actual power delivery across regions.  This is caused by the lack of real-time information necessary for scheduling and operation constraints.   For example, the information used in CTS for interchange scheduling is 75 minutes prior to the actual power delivery. With increasing integration of renewables,  interchange scheduling needs to be cognizant of uncertainty that arises between the time of interchange scheduling and that of power transfer. 

The goal of this paper is to obtain the optimal interchange schedule in the presence of system and operation uncertainty. To this end, we propose a two-stage stochastic optimization formulation aimed at minimizing the expected overall system cost.  The proposed optimization framework takes into account random fluctuations of load and renewable generations in the systems. Because directly solving the stochastic optimization is intractable, this paper presents an approach to transfer the stochastic optimization problem into an equivalent deterministic problem that maximizes the expected economic surplus. This transformation allows us to generalize the deterministic TO solution by intersecting expected demand and supply functions, therefore avoiding costly iterative computation between operators.

\subsection{Related Work}
There have been extensive studies on the seams issue. In this paper, we do not consider inefficiencies arise from market designs; we focus instead on optimizing the interchange schedule.    We highlight below approaches most relevant to the technique developed here.  For broadly related work, see \cite{IRIS,ConejoEtal06Springer_Decomposition, ZhaoLitvinovZheng14TPS,BaldickChatterjee14COR,LiEtal15TPS,ChenThorpMount04HICSS, IlicLang12} and  references therein.

Mathematically, optimal interchange scheduling can be obtained from the multi-area Optimal Power Flow (OPF) problem, which is a decentralized optimization of power flow that can be solved using various decomposition techniques \cite{ConejoEtal06Springer_Decomposition}. A general approach is based on the principle of Lagrangian Relaxation (LR) that decomposes the original problem into smaller subproblems. Some of the earliest approaches include the pioneer work of Kim and Baldick \cite{KimBaldick97TPS} and Conejo and Aguado \cite{ConejoAguado98TPS_MOPF} that predate the broad deregulation of the electricity market in the US.  Multi-area OPF problems that explicitly involve multiple ISOs have been widely studied \cite{CadwaladerHarveyPopeHogan98}\cite{ChenThorpMount04HICSS}.  In general, decentralized OPF based techniques typically require iterations between ISOs where one control center uses intermediate solutions from the other and solves its own dispatch problem.   Although the convergence of such techniques is often guaranteed under the DC-OPF formulation, the number of iterations can be large and the practical cost of communications and computations substantial.  We note that the recent marginal decomposition technique \cite{ZhaoLitvinovZheng14TPS} is shown to converge in a finite (but unknown) number of iterations.  

The growth of renewable integration has brought new attention to uncertainty in seams. Both stochastic optimization and robust optimization approaches have been considered recently.  In particular, Ahmadi-Khatir, Conejo, and Cherkaoui formulate a two-stage stochastic market clearing model in \cite{AhmadiConejoCherkaoui13TPS} for the multi-area energy and reserve dispatch problem.  The solution to the stochastic optimization is obtained based on scenario enumerations, which requires a prohibitively high computation effort.  In \cite{LiEtal15TPS}, the day-ahead tie-flow scheduling is considered in the unit commitment problem under wind generation uncertainty. Specifically, a two-stage adaptive robust optimization problem is formulated with the goal of minimizing the cost of the worst-case wind scenario, and solved by the column-and-constraint generation algorithm. The present paper complements these existing results by focusing on the  real-time interchange scheduling and develop a tractable stochastic optimization technique.

A pragmatic approach to the seams problem, one that has been adopted in practice and that can incorporate external market participants, is the use of proxy buses representing the interface between neighboring regions.  The technique presented here falls into this category.  Among existing prior work is the work by Chen et. al. \cite{ChenThorpMount04HICSS} where a coordinated interchange scheduling scheme is proposed for the co-optimization of energy and ancillary services. The technique is based on (augmented) LR involving iterations among neighboring control centers.  The work closest to ours are the TO technique presented in \cite{IRIS} and the work of Ilic and Lang \cite{IlicLang12}.  The underlying principle of \cite{IRIS} and \cite{IlicLang12} is based on the economics argument of supply and demand functions, which are exchanged by the neighboring operators.   For the (scalar) net interchange, such functions can be succinctly characterized, and the exchange needs to be made only once; the need of iterations among control centers is eliminated.  Our approach is also based on the same economics argument with the innovation on incorporating system and operation uncertainty.  Note that this type of approaches do not solve the multi-area OPF problem except the special case when there is a single tie line connecting the two operating regions.

\section{Deterministic Interchange Scheduling}\label{sec:iris}
\subsection{Proxy Bus Representation}
In practice, coordination between neighboring control regions and markets are typically through the use of proxy bus mechanism.  As pointed out in \cite{Harvey03Proxy}, a proxy bus models the location at which marginal changes in generation are assumed to occur in response to changes in inter-regional transactions. The proxy bus mechanism is utilized by all of the existing LMP based markets for representing and valuing interchange power \cite{Harvey03Proxy}. 

In this paper, we consider a power system consisting of two independently operated subsystems, as illustrated in Figure \ref{fig:sys1}.  Each operator selects a proxy bus to represent the location of import or export in the neighboring region. Specifically, as shown in Figure \ref{fig:proxy}, the operator from region 1 assumes a withdrawal $q$ at proxy bus $p_1$ and the operator from region 2 assumes an injection with the same quantity $q$ at proxy bus $p_2$.

\begin{figure}\centering
\includegraphics[scale=0.42]{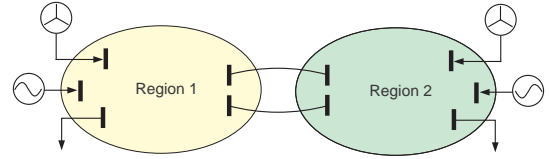}
\caption{A 2-region system with an interface.}\vspace{-1em} \label{fig:sys1}\end{figure}
\begin{figure}\centering\begin{psfrags}
\psfrag{t}[l]{\footnotesize {$q$}}
\psfrag{p}[c]{\scriptsize $p_1$}
\psfrag{q}[c]{\scriptsize $p_2$}
\includegraphics[scale=0.42]{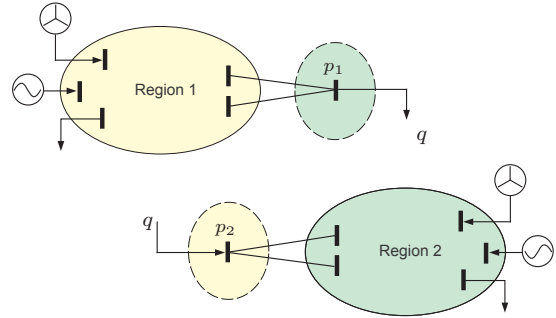}
\end{psfrags}\caption{The single proxy bus representation.}\vspace{-1em} \label{fig:proxy}\end{figure}

The interchange scheduling is to determine the value of the net interchange\footnote{The net interchange between two neighboring regions is the total amount of power flowing from one operating region to another.} $q$ that minimizes the overall operating cost subject to generation and transmission constraints. Note that, except when there is a single tie line that connects  two regions,  the proxy representation is an approximation, and the optimal interchange scheduling based on the proxy representation does not provide the optimal interchange of the original system.  In general, the optimal interchange via proxy representation is strictly suboptimal when it is compared with multi-area OPF solutions.

\subsection{Optimal Interchange Scheduling}
The interchange scheduling problem under the proxy bus model can be formulated as minimizing the generation costs of both regions with respect to the power balance, transmission (internal and interface) and generator constraints. For simplicity, we make the following assumptions throughout the paper: (i) the system is lossless, and (ii) the cost function $c_i(\cdot), i\in\{1,2\}$, is quadratic in the form $c_i(\cdot)=g_i^\intercal H_i g_i +q_i g_i$ where matrix $H_i$ is positive definite. Under the single proxy bus system, the net interchange can be modeled explicitly as an additional scalar variable in the optimization problem as follows:

\begin{equation*}\label{problem:min_cost}
\setlength\arraycolsep{2pt}\begin{array}{r l l}
(P_1)\quad \underset{q,g_1,g_2}\min &c_1(g_1)+c_2(g_2)&\\
\text{subject to} &\mathbf{1}^\intercal (d_1-g_1)+ q=0, & (\lambda_1)\\
&\mathbf{1}^\intercal (d_2-g_2)-q=0, & (\lambda_2)\\
&S_{1}(d_1-g_1)+S_{q1} q \leq F_1, &(\mu_1)\\
&S_{2}(d_2-g_2)-S_{q2} q \leq F_2, &(\mu_2)\\
&q \leq Q, & (\mu_{q})\\
&g_1\in \mathcal{G}_1,&\\
&g_2 \in \mathcal{G}_2,&
\end{array}
\end{equation*}
where

\hspace{-1em}\begin{tabular}{p{.3cm} p{7.9cm}}
$c_i(\cdot)$& real-time generation offer function for region $i$;\\
$d_i$&  vector of forecasted load and renewable generation for region $i$;\\
$q$&  net interchange from region $1$ to region $2$, if $q>0$; from region $2$ to region $1$, otherwise;\\
$g_i$&  vector of dispatches for region $i$;\\
$F_i$& vector of transmission limits for region $i$;\\
$Q$& interface limit;\\
$\mathcal{G}_i$& generator constraints for region $i$;\\
${S}_{ij}$&  shift factor matrix of buses in region $i$ to transmission lines in region $j$;\\
${S}_{qi}$& shift factor vector of buses in region $i$ to the interface;\\
$\lambda_i$& shadow price for power balance constraint in region $i$;\\
$\mu_i$& shadow prices for transmission constraints in region $i$;\\
$\mu_{q}$& shadow price for the net interchange constraint.
\end{tabular}\vspace{.2em}

The problem $(P_1)$ is a centralized formulation for determining the optimal interchange between region 1 and 2. Such an optimization problem requires a coordinator who have full access to all related information of both regions which is unsuitable in the present deregulated electricity markets. 

As in \cite{BaldickChatterjee14COR}, the centralized problem $(P_1)$ can be written in a hierarchical form of decentralized optimization as follows.
\begin{equation*}\label{problem:min_cost_m}
\setlength\arraycolsep{2pt}\begin{array}{r l l}
(P_2)\quad\underset{q}\min &c_1(g^*_1(q))+c_2(g_2^*(q))&\\
\text{subject to} &q \leq Q, & (\mu_{q})
\end{array}
\end{equation*}
where $g^*_i(q)$, $i\in\{1,2\}$, is the optimal dispatch for region $i$, given the interchange level $q$.

The regional dispatch problem for region $1$ is specified as
\begin{equation*}\label{problem:min_cost_a}
\setlength\arraycolsep{2pt}\begin{array}{r l l}
(P_{21})\quad \underset{g_1\in\mathcal{G}_1}\min &c_1(g_1)&\\
\text{subject to} &\mathbf{1}^\intercal (d_1-g_1)+q=0, & (\lambda_1)\\
&S_{1}(d_1-g_1)+S_{q1} q \leq F_1, &(\mu_1)\\
\end{array}
\end{equation*}
and for region $2$ as
\begin{equation*}\label{problem:min_cost_b}
\setlength\arraycolsep{2pt}\begin{array}{r l l}
(P_{22})\quad \underset{g_2\in\mathcal{G}_2}\min &c_2(g_2)&\\
\text{subject to} &\mathbf{1}^\intercal (d_2-g_2)- q=0, & (\lambda_2)\\
&S_{2}(d_2-g_2)-S_{q2} q \leq F_2. &(\mu_2)\\
\end{array}
\end{equation*}

Note that the optimization problem involves an outer problem $(P_2)$ to optimize the interchange level $q$, and an inner problem that is naturally decomposed into two regional problems, both parameterized by $q$. In other words, the optimizer and the associated Lagrangian multipliers for $(P_{21})$ and $(P_{22})$ are functions of $q$, \textit{i.e.}, $g^*_i(q)$, $\lambda^*_i(q)$, $\mu^*_i(q)$, $i \in \{1, 2\}$.

\subsection{Tie Optimization}
The key idea of TO is to determine the interchange schedule by intersecting the demand and supply curves. By interchange demand/supply curve we mean the incremental cost of the regional dispatch at the interface, which is essentially the LMP at the proxy bus. Given the interchange level $q$, the LMP at the proxy bus for region $i\in \{1, 2\}$ is defined as
\begin{equation}\label{def:pi}
\pi_i(q)=\lambda^*_i(q)+{S}_{qi}^\intercal \mu^*_i(q),
\end{equation}
where $\lambda^*_i(q)$ and  $\mu^*_i(q)$ are the Lagrangian multipliers associated with the optimal solution of $(P_{2i})$, for $i \in \{1, 2\}$. Note that function $\pi_i(q)$ is ISO $i$'s incremental dispatch cost at the interface at the net interchange $q$, which serves as a supply curve for the exporting ISO or a demand curve for the importing ISO. 

\begin{figure}\centering
\begin{psfrags}
\psfrag{Sa}[l]{\scriptsize $\pi_2(q)$}
\psfrag{Sb}[l]{\scriptsize $\pi_1(q)$}
\psfrag{p}[r]{\scriptsize $\pi^*$}
\psfrag{A}[c]{}
\psfrag{B}[c]{}
\psfrag{Price}[c]{\scriptsize $\pi$}
\psfrag{flow}[l]{\scriptsize $q$}
\psfrag{lim}[l]{\scriptsize $Q$}
\psfrag{f}[l]{\scriptsize {${q_\text{TO}^*}$}}
\includegraphics[width=.3\textwidth]{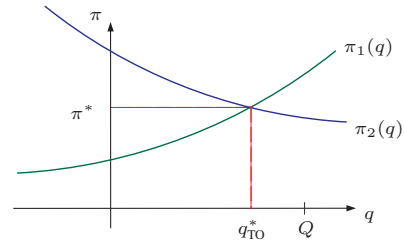}
\end{psfrags}\caption{Illustration of TO.}\vspace{-1em}\label{fig:to}\end{figure}

We use the graphical representation in \cite{IRIS} to illustrate the basic principle of TO. As shown in Figure \ref{fig:to}, $\pi_i(q)$ represents the generation supply curve for region $i$, but $\pi_2(q)$ is drawn in a descending cost order.  In this example, the direction of interface flow\footnote{The direction of interface flow can be determined by comparing the prices at $q=0$: if $\pi_1(0)<\pi_2(0)$, the power flows from region 1 to 2; otherwise, the direction of interface flow is opposite. } is from region 1 to region 2; $\pi_1(q)$ and $\pi_2(q)$ serve as the supply and demand curve respectively. The optimal  schedule $q_\text{TO}^*$ is set at the intersection of the two curves. Note that if this quantity exceeds the interface capacity $Q$, the schedule should be set at the maximum capacity instead. The interface transmission constraint, in this case, becomes binding and price separation happens between markets. It should also be noted that import or export transactions are settled at the {\em real-time LMP} which is calculated at the proxy bus after the delivery.

According to \cite{IRIS}, the interchange schedule of TO is the optimal solution to $(P_1)$ (as well as $(P_2)$). This intuitive argument is a manifestation of a deeper connection between social welfare optimization illustrated in Figure \ref{fig:to} and cost minimization defined by $(P_1)$.  In what follows, we will exploit this connection in the presence of uncertainty. 
\section{A Stochastic Interchange Scheduling}\label{sec:siris}
So far, we have described the interchange scheduling in a deterministic system setting. We now focus on the incorporation of random load and generation in the scheduling scheme.
\subsection{Stochastic Programming Formulation}
Stochastic optimization is the most common framework to model optimization problems involving uncertainty.  Consider the case that load (or stochastic generation, treated as a negative load) is random. The inter-regional interchange scheduling can be formulated as a two-stage stochastic optimization problem. The first stage involves optimizing the net interchange to minimize the expected overall cost
\begin{equation*}\label{problem:min_cost_m2}
\setlength\arraycolsep{2pt}\begin{array}{r l l}
(P_4)\quad\underset{q}\min &\mathbb{E}_{d_1, d_2}\left[c_1(g^*_1(q,d_1))+c_2(g_2^*(q,d_2))\right]&\\
\text{subject to} &q \leq Q, & (\mu_{q})
\end{array}
\end{equation*}
and the second stage solves the regional optimal dispatch problem given the interchange level $q$ and the realization of random load $d_1$ and $d_2$, which are specified as $(P_{21})$ and $(P_{22})$. Note that the optimal dispatch and the associated Lagrangian multipliers of $(P_{2i})$ are parameterized by two factors: the interchange level $q$ and the load realization $d_i$. So the LMP $\pi_i(q, d_i)$ at the proxy bus is a function of both $q$ and $d_i$.

Directly solving this problem requires the distribution of the regional cost function $c_i(q,d_i)$ at each interchange level, and  a coordinator to determine the optimal schedule, neither of which is achievable in the present deregulated electricity markets. In general, this two-stage stochastic optimization problem is intractable, especially when the load and renewable generation forecast follows a continuous distribution. The proposed scheduling technique, on the other hand, can solve this problem without increasing the computation complexity of deterministic TO. Details are provided in the next two subsections.

\subsection{Social Welfare Optimization}
The main idea of solving $(P_4)$ is to exploit the connection between cost minimization and social welfare optimization under uncertainty.  With the randomness present in the second stage of $(P_4)$, it is not obvious how the two-stage optimization problem can be transformed into a corresponding form of social welfare optimization.  It turns out that the optimal solution can be obtained by solving a deterministic TO problem using the {\em expected} demand and supply functions.

We now present an optimization problem from the import-export perspective, but taking into account that both import and export regions must agree on the forward interchange quantity in the presence of future demand and supply uncertainty.  Because the interchange quantity is fixed ahead of the actual power delivery, each region may have to rely on its internal resources to compensate uncertainty in real time.  To this end, it is reasonable for the export region to maximize its {\em expected producer surplus} and the import region to maximize its {\em expected consumer surplus}. 

\begin{figure}\centering
\begin{psfrags}
\psfrag{Sa}[l]{\scriptsize $\mathbb{E}_{d_2}[\pi_2(q,d_2)]$}
\psfrag{Sb}[l]{\scriptsize $\mathbb{E}_{d_1}[\pi_1(q,d_1)]$}
\psfrag{p}[r]{\scriptsize $\pi^*$}
\psfrag{Price}[c]{\scriptsize $\pi$}
\psfrag{flow}[l]{\scriptsize $q$}
\psfrag{lim}[l]{\scriptsize $Q$}
\psfrag{f}[l]{\scriptsize {${q_\text{STO}^*}$}}
\psfrag{A}[l]{}
\psfrag{B}[r]{}
\psfrag{S}[l]{}
\psfrag{SS}[c]{\scriptsize Producer surplus}
\psfrag{C}[c]{\scriptsize Consumer surplus}
\includegraphics[width=.3\textwidth]{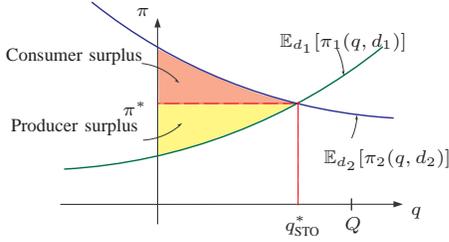}
\end{psfrags}\caption{Illustration of STO.}\vspace{-1em}\label{fig:sto}\end{figure}

Without loss of generality, let region 1 be the exporter.  For fixed interchange $q$, let $\pi_i(q,d_i)$ be the (random)  LMP at the proxy bus.  Then $\mathbb{E}_{d_1}[\pi_1(q,d_1)]$, as a function of interchange $q$,  is the expected supply curve averaged over its internal randomness.  Similarly, $\mathbb{E}_{d_2}[\pi_2(q,d_2)]$ is the expected demand curve averaged over the internal randomness in region 2 at the time of delivery.  

As shown in Figure \ref{fig:sto}, the optimal interchange quantity $q_\text{STO}^*$ that maximizes the expected social welfare (in absence of interface constraint) is simply the intersection of the expected demand and supply curves. In general, the interchange that maximizes the expected social welfare is given by

\begin{equation*}\label{problem:max_surplus_sto}
\setlength\arraycolsep{2pt}\begin{array}{r l l}
(P_5)\quad\underset{q}\max &\int_0^{q} \mathbb{E}_{d_2}[\pi_2(x, d_2)]-\mathbb{E}_{d_1}[\pi_1(x, d_1)]dx&\\
\text{subject to} &q \leq Q.&(\mu_{q})\\
\end{array}
\end{equation*}

To solve $(P_5)$, each operator needs to compute, for each interchange quantity $q$,  the expected LMP at its own proxy bus.   Such a computation requires the conditional expectation of future LMP at the time of delivery.  The conditional expectation can be obtained through probabilistic LMP forecast using models for load and generation. See, for example, \cite{JiTongThomas15ARXIV}. It is also conceivable that the conditional expectations can also be approximated via regression analysis.  Once the expected demand and supply functions are obtained, solving for the optimal interchange quantity becomes a one-dimensional search.  

\subsection{Stochastic Tie Optimization}
In this section, we establish formally the equivalence of $(P_4)$ and $(P_5)$ where the solution of $(P_5)$ solves the stochastic optimization problem $(P_4)$.
\begin{thm}\label{thm:eqv_s}
If the optimal dual solutions of $(P_{21})$ and $(P_{22})$ are unique for all $q\leq Q$,  then $(P_4)$ and $(P_5)$  are equivalent in the sense that they have the same optimizer.
\end{thm}

Theorem \ref{thm:eqv_s} provides a new way, we call it Stochastic Tie Optimization (STO), to solve the intractable problem $(P_4)$. This result is significant because the optimal interchange can be obtained from a deterministic optimization problem $(P_5)$ which only requires the information of the expected supply and demand curves. Since these price functions are non-confidential information, $(P_5)$ can be solved by one of the operators if the other operator shares its price curve. In this way, operators do not need to iteratively update or exchange information within the scheduling procedure. This property is in contrast to most decomposition methods where subproblems are resolved and intermediate results are exchanged in each iteration. Because one-time information exchange is sufficient for the optimal schedule, operators do not need to repeatedly solve regional OPF, which is computationally expensive for sufficiently large systems, within the scheduling procedure. Such a property significantly reduces the computation costs in real time, thereby providing the potential of higher scheduling frequency.

Now we provide the proof of Theorem \ref{thm:eqv_s}.
\begin{proof}[Proof of Theorem \ref{thm:eqv_s}]
 We first show the differentiability of the objective functions of $(P_4)$ and $(P_5)$. This follows immediately from the well known results in multiparametric quadratic programming summarized in Lemma \ref{lem:mqp}.
\begin{lem}[\hspace{-.1em}\cite{mpc_book}] \label{lem:mqp}
If the dual problem of $(P_{2i})$, $i\in\{1,2\}$, is not degenerate for all $q\leq Q$, then 
\begin{enumerate}
\item the optimizer of $(P_{2i})$ and associated vector of Lagrangian multipliers are continuous and piecewise affine (affine in each critical region), and
\item the optimal objective of $(P_{2i})$ is continuous, convex and piecewise quadratic (quadratic in each critical region).
\end{enumerate}
\end{lem}

By Lemma \ref{lem:mqp}, the objective function of $(P_4)$, denoted by $J(q)$, is differentiable with derivative 
\setlength\arraycolsep{2pt}\begin{eqnarray*}J'(q)&=&\mathbb{E}_{d_1,d_2}\left[\frac{\partial}{\partial q}c_1(g_1^*(q,d_1))+\frac{\partial}{\partial q}c_2(g_2^*(q,d_2))\right]\\
&=&\mathbb{E}_{d_1}[\pi_1(q,d_1)]-\mathbb{E}_{d_2}[\pi_2(q,d_2)]
\end{eqnarray*}
where the second equality holds by the Envelope Theorem. Lemma \ref{lem:mqp} also implies that $\pi_1(q,d_1)$ and $\pi_2(q,d_2)$ are continuous functions, so the objective function of $(P_5)$, denoted by $s(q)$, is differentiable with derivative
\[s'(q)=\mathbb{E}_{d_1}[\pi_1(q^\sharp,d_1)]-\mathbb{E}_{d_2}[\pi_2(q^\sharp,d_2)].\]

Now we derive the connection between the optimal solutions to $(P_4)$ and $(P_5)$ from the first order conditions. The optimal solution $q^{*}$ to $(P_4)$ and the associated Lagrangian multiplier $\mu^{*}_{q}$ satisfy the first order condition for $(P_4)$:
\begin{equation}\label{eqn:foc}
\mathbb{E}_{d_1}[\pi_1(q^{*},d_1)]-\mathbb{E}_{d_2}[\pi_2(q^{*},d_2)]+\mu^{*}_{q}=0.
\end{equation}
Similarly, the optimal solution $q^\sharp$ to $(P_5)$ and  the associated Lagrangian multiplier $\mu^\sharp_{q}$ satisfy the first order condition for $(P_5)$
\begin{equation}
\mathbb{E}_{d_1}[\pi_1(q^\sharp,d_1)]-\mathbb{E}_{d_2}[\pi_2(q^\sharp,d_2)]+\mu^\sharp_{q}=0,
\end{equation}
which is exactly the same as (\ref{eqn:foc}).

Finally, we show $q^{*}=q^\sharp$. To prove this, we need the monotonicity of price function $\pi_i(q)$ (with fixed $d_i$ as defined in (\ref{def:pi})) which is summarized in the following lemma whose proof is provided in the appendix.
\begin{lem}\label{lem:mono}
If the dual problem of $(P_{2i}), i\in\{1,2\}$, has a unique optimal solution for all $q\leq Q$,  then $\pi_1(q)$ is monotonically increasing and $\pi_2(q)$ is monotonically decreasing.
\end{lem}
 Below we show that in each of the following cases, either the case itself is impossible or $q^{*}=q^\sharp$.
\begin{enumerate}
  \item $q^{*}=q^\sharp=Q$. The statement is trivially true.
  \item $q^{*}<Q$ and $q^\sharp <Q$.  In this case, the interface constraint is not binding in either problem, so we have $\mu_{q}^{*}=\mu_{q}^\sharp=0$, which implies that $\mathbb{E}_{d_1}[\pi_1(q^{*},d_1)]=\mathbb{E}_{d_2}[\pi_2(q^{*},d_2)]$ and $\mathbb{E}_{d_1}[\pi_1(q^\sharp,d_1)]=\mathbb{E}_{d_2}[\pi_2(q^\sharp,d_2)]$. By Lemma \ref{lem:mono} and the preservation of monotonicity under expectation operation, there is a unique solution to
$\mathbb{E}_{d_2}[\pi_2(q,d_2)]-\mathbb{E}_{d_1}[\pi_1(q,d_1)]=0$. Therefore, $q^{*}=q^\sharp$.
  \item  $q^{*}<q^\sharp =Q$. We construct a solution of $(P_1)$ using  $q^\sharp$ and the associated optimal functions defined in $(P_{21})$ and $(P_{22})$, \textit{i.e.}, $q^\sharp$, $\mu_{q}^\sharp$ (which is zero),  $g_1^*(q^{\sharp})$, $g_2^*(q^{\sharp})$,  $\lambda_1^*(q^{\sharp})$, $\mu_1^*(q^{\sharp}),\lambda_2^*(q^{\sharp}), \mu_2^*(q^{\sharp})$. Note that this solution satisfies the first order conditions for $(P_1)$, so it is optimal to $(P_1)$. However, this contradicts the uniqueness of the optimizer to $(P_1)$. Therefore, this case is impossible. 
  \item   $q^\sharp <q^{*}=Q$. This case is also impossible. The proof follows the logic of that in case 3).
\end{enumerate}

 To sum up, $(P_4)$ and $(P_5)$ are equivalent in the sense that they share the same optimal solution.
\end{proof}

\section{Stochastic CTS}\label{sec:cts}
In this section, we incorporate external market participants in STO, which generalizes the CTS proposal currently in implementation.  This generalization, we call it Stochastic Coordinated Transaction Scheduling (SCTS), is simply replacing the supply and demand curves used in CTS  by their expected values. 

As in CTS, market participants are allowed to submit requests to buy and sell power simultaneously on each side of the interface. Such request is called \textit{interface bid}, which includes a price indicating the minimum expected price difference between the two regions that the participant is willing to accept, a transaction quantity and its direction. 

We use a similar graphical representation of STO to illustrate the scheduling procedure of SCTS. As shown in Figure \ref{fig:scts}, $\mathbb{E}_{d_1}[{\pi}_1(q,d_1)]$ is the expected supply curve of region 1, and $\mathbb{E}_{d_2}[\tilde{\pi}_2(q,d_2)]$ is the adjusted curve of $\mathbb{E}_{d_2}[\pi_2(q,d_2)]$ by subtracting the aggregated interface bids $\pi_\text{bid}(q)$. The SCTS schedule is set at the intersection of $\mathbb{E}_{d_1}[\pi_1(q,d_1)]$ and $\mathbb{E}_{d_2}[\tilde{\pi}_2(q,d_2)]$. All the interface bids to the left of $q^*_\text{SCTS}$ are accepted and settled at the real-time LMP difference.

\begin{figure}\centering\begin{psfrags}
\psfrag{Sa}[l]{\scriptsize $\mathbb{E}_{d_2}[\pi_2(q,d_2)]$}
\psfrag{Sa2}[l]{\scriptsize $\mathbb{E}_{d_2}[\tilde{\pi}_2(q,d_2)]$}
\psfrag{Sb}[l]{\scriptsize $\mathbb{E}_{d_1}[\pi_1(q,d_1)]$}
\psfrag{p}[c]{\scriptsize $\pi^*$ }
\psfrag{f}[l]{\scriptsize {${q^*_\text{SCTS}}$}}
\psfrag{Price}[c]{\scriptsize $\pi$}
\psfrag{flow}[l]{\scriptsize $q$}
\psfrag{lim}[l]{\scriptsize $Q$}
\psfrag{A}[l]{}
\psfrag{B}[l]{}
\includegraphics[width=.3\textwidth]{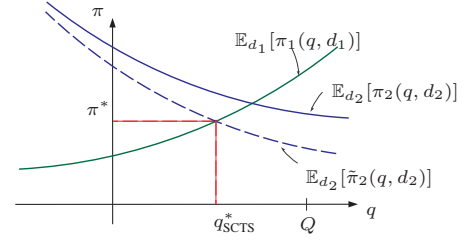}
\end{psfrags}\caption{Illustration of SCTS.}\label{fig:scts}\vspace{-1em}\end{figure}

The scheduling and clearing procedure described above is summarized as follows:
\begin{enumerate}
  \item  share the expected LMP functions $ \mathbb{E}_{d_1}[\pi_1(q,d_1)]$ and $ \mathbb{E}_{d_2}[\pi_2(q,d_2)]$;
  \item determine the direction of the interchange flow by comparing $\mathbb{E}_{d_2}[\pi_2(0,d_2)]$ and $\mathbb{E}_{d_1}[\pi_1(0,d_1)]$;
  \item construct the aggregated interface bid curve $\pi_\text{bid}(q)$ which is a stack of all interface bids with the direction determined in step 2) in an increasing order of the submitted price difference;
  \item calculate the optimal SCTS schedule from the following optimization problem $(P_6)$.
\end{enumerate}
  \[\setlength\arraycolsep{2pt}\begin{array}{r l}
(P_6)\quad\underset{q}\max &\int_0^{q} \mathbb{E}_{d_2}[\pi_2(x, d_2)]-\mathbb{E}_{d_1}[\pi_1(x, d_1)]-\pi_\text{bid}(x)dx\\
\text{subject to} &q \leq Q.\\
\end{array}\]

Note that the only difference between STO and SCTS is the inclusion of interface bids. All other components are identical.  This implies that one-time information exchange is sufficient; no iteration between operators is necessary during the scheduling procedure when one operator submits its expected generation supply curve to the other who executes this scheduling and clearing procedure.

\section{Evaluation}\label{sec:sim}
In this section, we compare the performance of the proposed STO with that of TO on two systems: a 6-bus system and the IEEE 118-bus system.  In particular, we focus on the two most common symptoms of seams: (i) the under-utilization of interface transmission, and (ii) the presence of counter-intuitive flows from the high cost region to the low cost region. In both examples, TO uses the certainty equivalent forecast of the stochastic generation, \textit{i.e.}, the mean value, while STO uses the probabilistic forecast, \textit{i.e.}, the distribution. Various scenarios are studied in these two examples.

\subsection{Example 1: a 2-region 6-bus system} Consider a 2-region 6-bus system as depicted in Figure \ref{fig:6bus}. Generator incremental cost functions, capacity limits, and load levels (the default values) are presented in the figure. All lines are identical except for the maximum capacities: the tie lines (line 2-6 and line 3-4) and the internal transmission lines in region 1 have the maximum capacities of $100$ MW, and the internal lines in region 2 have the maximum capacities of $200$ MW. The system randomness comes from the wind generator at bus 1 in region 1. The entire network model (the shift factor matrix) is assumed to be known to both ISOs. By default, we chose bus 3 as the proxy bus to represent the network in region 1, and bus 6 to represent the network in region 2. The impact of the location of proxy buses will be further investigated.
\begin{figure}\hspace{1em}\begin{psfrags}
\psfrag{1}[c]{\scriptsize 1}
\psfrag{2}[c]{\scriptsize 2}
\psfrag{3}[c]{\scriptsize 3}
\psfrag{4}[c]{\scriptsize 4}
\psfrag{5}[c]{\scriptsize 5}
\psfrag{6}[c]{\scriptsize 6}
\psfrag{A}[c]{\scriptsize Region 1}
\psfrag{B}[c]{\scriptsize Region 2}
\psfrag{l1}[c]{\tiny $d_2=30$}
\psfrag{l2}[c]{\tiny $d_5=250$}
\psfrag{g1}[c]{\tiny$\renewcommand{\arraystretch}{1.5}\begin{array}{c} 0\leq g_1\leq 120\\c_1=0.01g^2_1+10g_1 \\ \end{array}$}
\psfrag{g3}[c]{\tiny$\renewcommand{\arraystretch}{1.5}\begin{array}{c} 0\leq g_3\leq 200\\c_3=0.01g^2_3+40g_3\\  \end{array}$}
\psfrag{g4}[c]{\tiny$\hspace{1em}\renewcommand{\arraystretch}{1.5}\begin{array}{c} 0\leq g_4\leq 100\\c_4=0.01g^2_4+30g_4\\  \end{array}$}
\psfrag{g6}[c]{\tiny$\renewcommand{\arraystretch}{1.5}\begin{array}{c} 0\leq g_6\leq 200 \\c_6=0.01g^2_6+45g_6\\ \end{array}$}
\includegraphics[width=.48\textwidth]{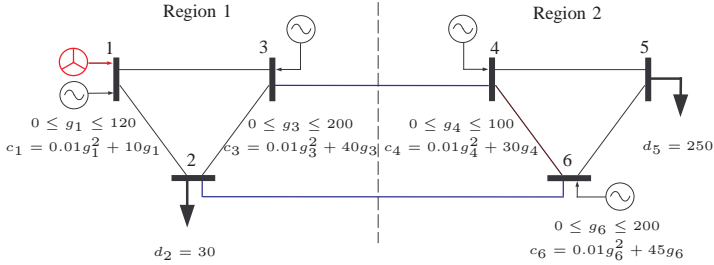}\end{psfrags}\caption{A 2-region 6-bus system.}\vspace{-1em}\label{fig:6bus}\end{figure}

\subsubsection{A baseline}
We first tested a baseline with the probabilistic wind forecast distribution $\mathcal{N}(55, 10^2)$. Two levels of load were chosen to illustrate the two symptoms of the inefficiency of TO schedule: the first load level $d_5=250$ is an example of the counter intuitive flow occurrence, and the second load level $d_5=200$ shows the case of the interface under utilization.  Results are presented in Figure \ref{fig:curve}-\ref{fig:pdiff} and Table \ref{table:base}.

Figure \ref{fig:curve} shows the generation supply curves of region 1 and region 2 under TO and STO for the two examples, respectively. $\pi_1^\text{TO}$ is the incremental cost of region 1 to deliver the power to the proxy bus 6 using the forecasted mean $55$ of the wind production, and $\pi_1^\text{STO}$ is the expected incremental cost using the forecast distribution $\mathcal{N}(55, 10^2)$. Since there is no randomness in region 2, the supply curves of region 2 for TO and STO are the same in both examples. At the interchange level of STO schedule, the expected overall system cost is minimized in both cases as shown in Figure \ref{fig:cost}, and the expected prices at the two proxy buses converge in both cases as shown in Figure \ref{fig:pdiff}. From Table \ref{table:base},  the expected price difference at the level of TO schedule in the first example is $-2.13$\$/MWh, which means that the expected price of the importing region (region 2) is lower then that of the exporting region. This implies that the interchange is scheduled from a high cost region to a low cost region, which is counter intuitive. On the other hand, in the second example, the expected price difference at the interchange level of TO schedule is $2.38$\$/MWh, \textit{i.e.}, the marginal price of the importing region is higher than that of the exporting region. With this price difference, increasing the interchange level can further reduce the expected overall cost, which implies the interchange capacity is under utilized. Because the interchange level of STO schedule is optimal as its design, any schedule more than this optimal level will cause the counter intuitive flow, and any schedule less than that will lead to the interface under utilization.

\begin{table}\caption{Comparison of TO and STO.}\label{table:base}\centering
\renewcommand{\arraystretch}{1.2}\begin{tabular}{|c|c|c|c|c|}\hline
Scenario&Method&$q^*$&$\mathbb{E}[\text{Cost}(q^*)]$&$\mathbb{E}[\Delta\pi(q^*)]$ \\ \hline
\multirow{2}{*}{$d_5=250$}&TO&$166.5$&$6794.2$&$-2.13$\\ \cline{2-5}
&STO&$162.8$&$6790.8$&$0$\\ \hline
\multirow{2}{*}{$d_5=200$}&TO&$147.5$&$4621.6$&$2.38$\\ \cline{2-5}
&STO&$151.4$&$4608.6$ &$0$ \\ \hline
\end{tabular}
\end{table}

\begin{figure}\centering
\begin{subfigure}{0.24\textwidth}
\includegraphics[width=\textwidth]{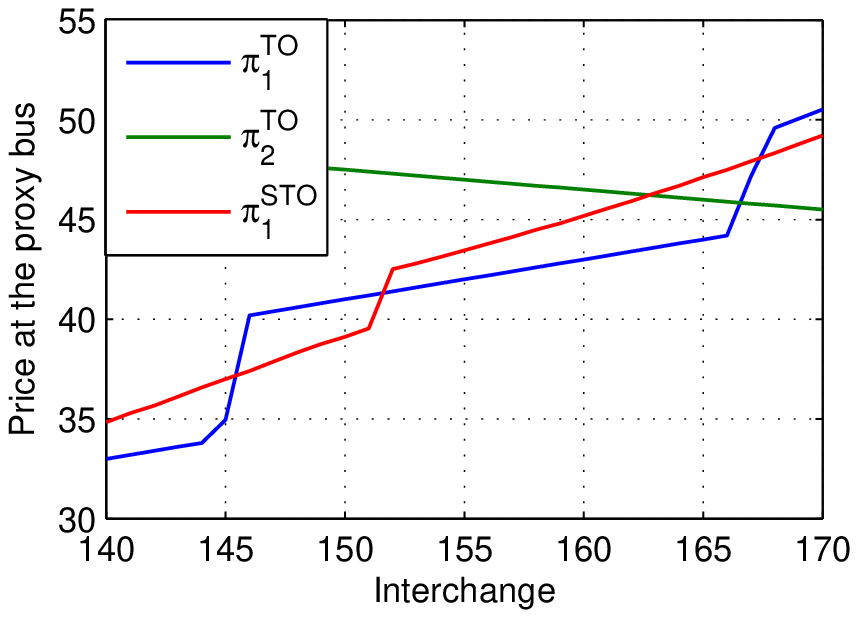}\caption{$d_5=250$.}
\end{subfigure}
\begin{subfigure}{0.24\textwidth}
\includegraphics[width=\textwidth]{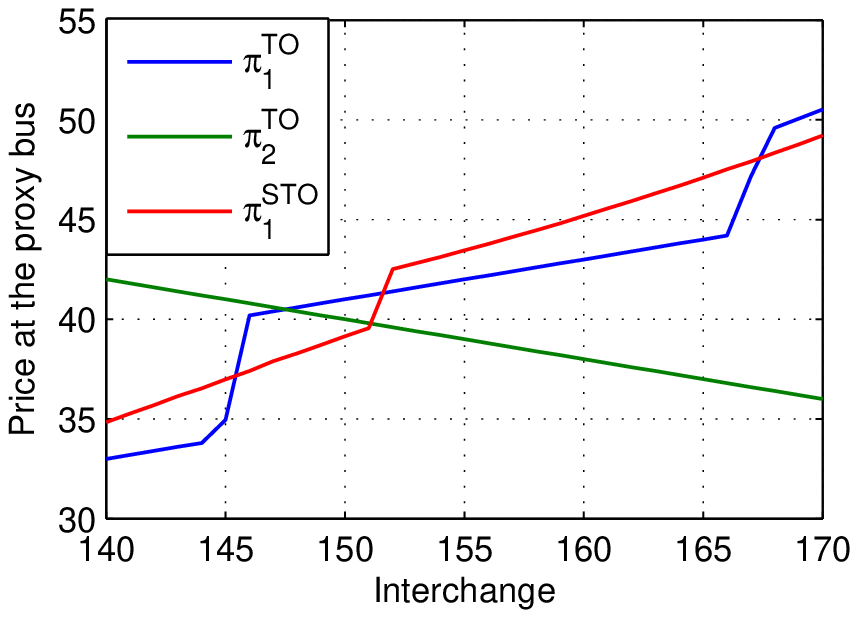}\caption{$d_5=200$.}
\end{subfigure}	
\caption{Generation supply curves.}\vspace{-1em}
\label{fig:curve}\end{figure}

\begin{figure}\vspace{-1em}\centering
\begin{subfigure}{0.24\textwidth}
\includegraphics[width=\textwidth]{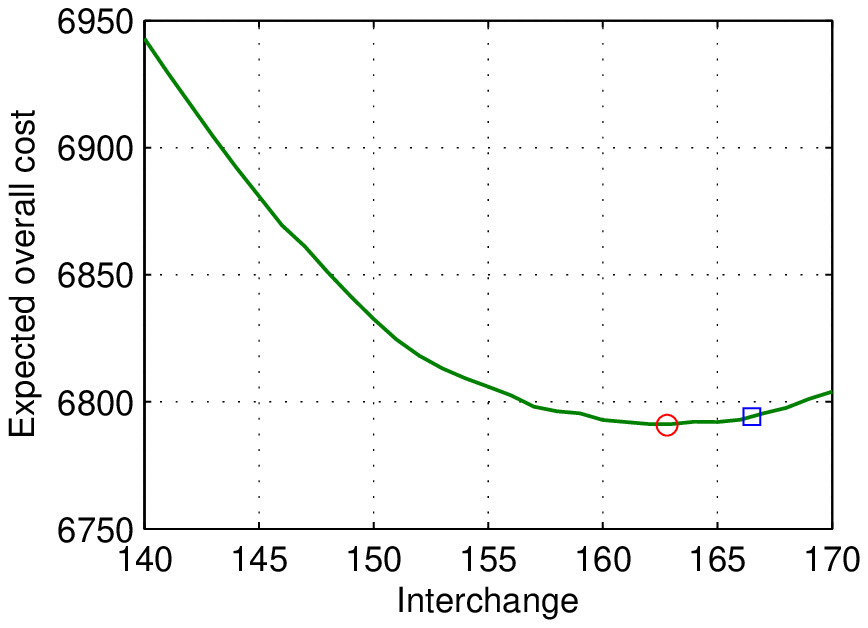}\caption{$d_5=250$.}
\end{subfigure}
\begin{subfigure}{0.24\textwidth}
\includegraphics[width=\textwidth]{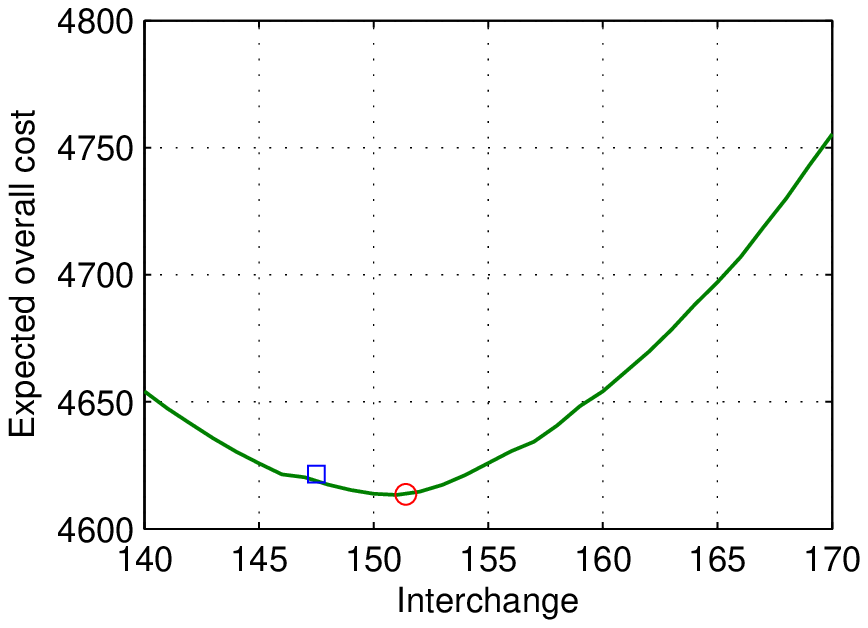}\caption{$d_5=200$.}
\end{subfigure}
\caption{Expected overall cost: TO is marked by the blue square and STO by the red circle.}
\label{fig:cost}\end{figure}

\begin{figure}\vspace{-1em}\centering	
\begin{subfigure}{0.24\textwidth}
\includegraphics[width=\textwidth]{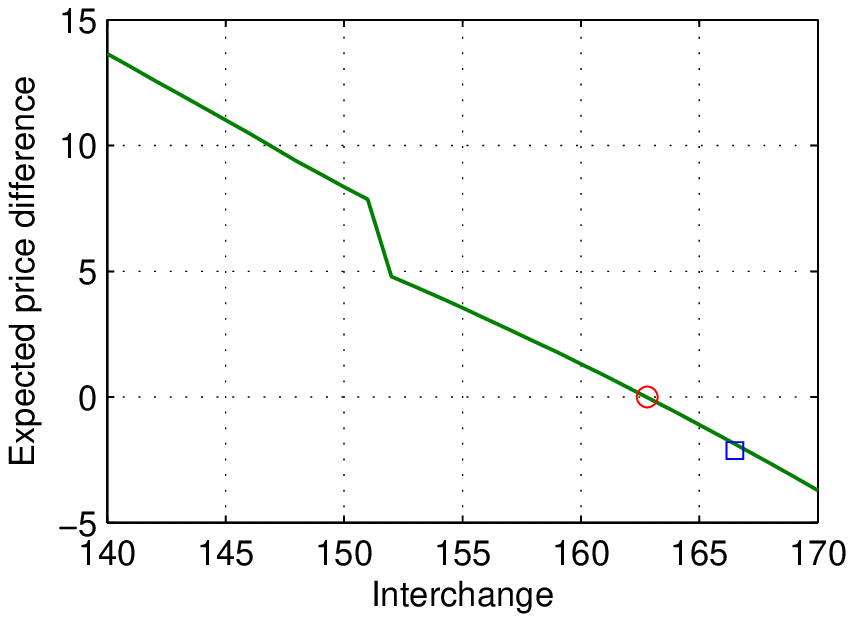}\caption{$d_5=250$.}
\end{subfigure}
\begin{subfigure}{0.24\textwidth}
\includegraphics[width=\textwidth]{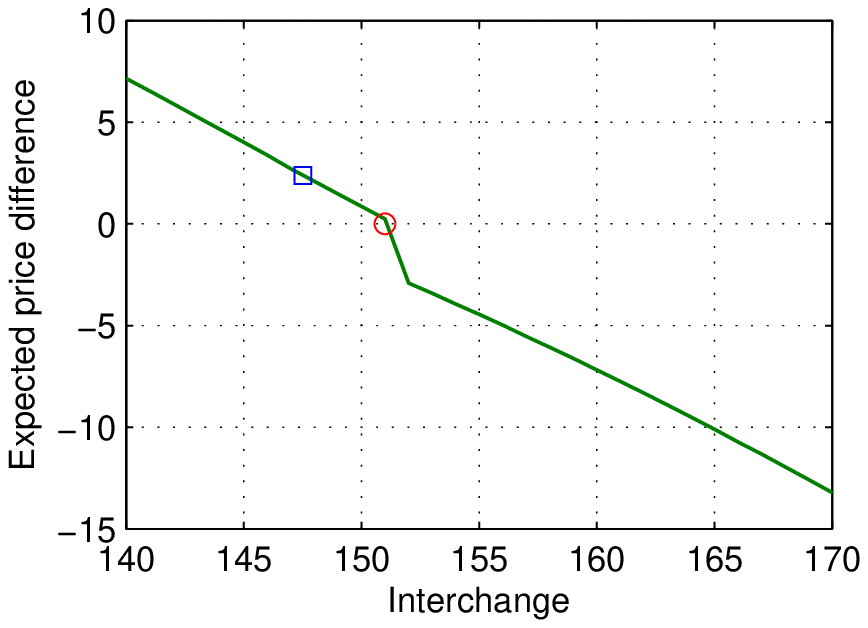}\caption{$d_5=200$.}
\end{subfigure}
\caption{Expected price difference: TO is marked by the blue square and STO by the red circle.}
\label{fig:pdiff}\end{figure}

\subsubsection{Impact of forecast uncertainty}
The impact of the forecast uncertainty level was then investigated by varying the standard deviation $\sigma$ of the probabilistic wind production forecast $w\sim \mathcal{N}(55, \sigma^2)$. Loads were set at the default values given in Figure \ref{fig:6bus}. Results are presented in Figure \ref{fig:sigma}.
\begin{figure}\centering
\begin{subfigure}{0.24\textwidth}
\includegraphics[width=\textwidth]{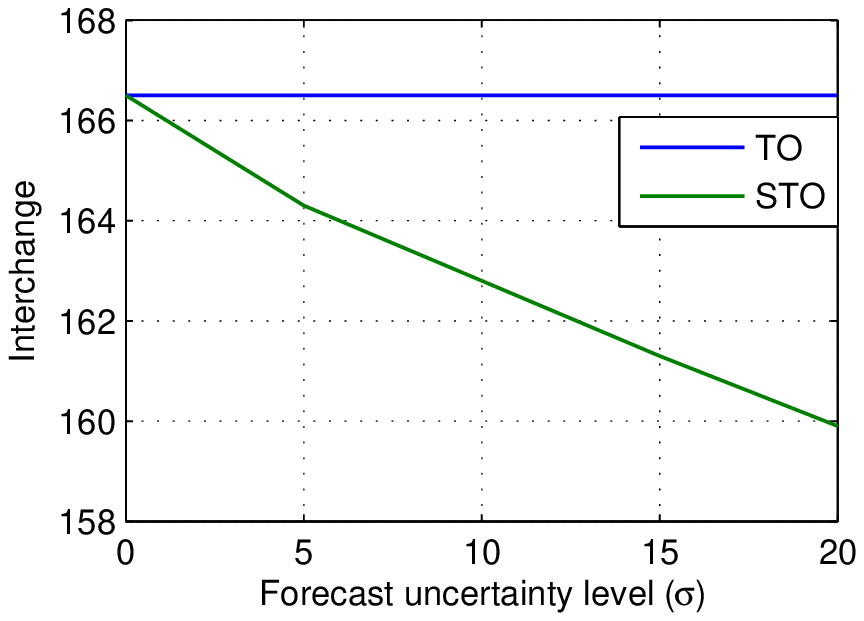}\caption{Interchange.}
\end{subfigure}
\begin{subfigure}{0.24\textwidth}
\includegraphics[width=\textwidth]{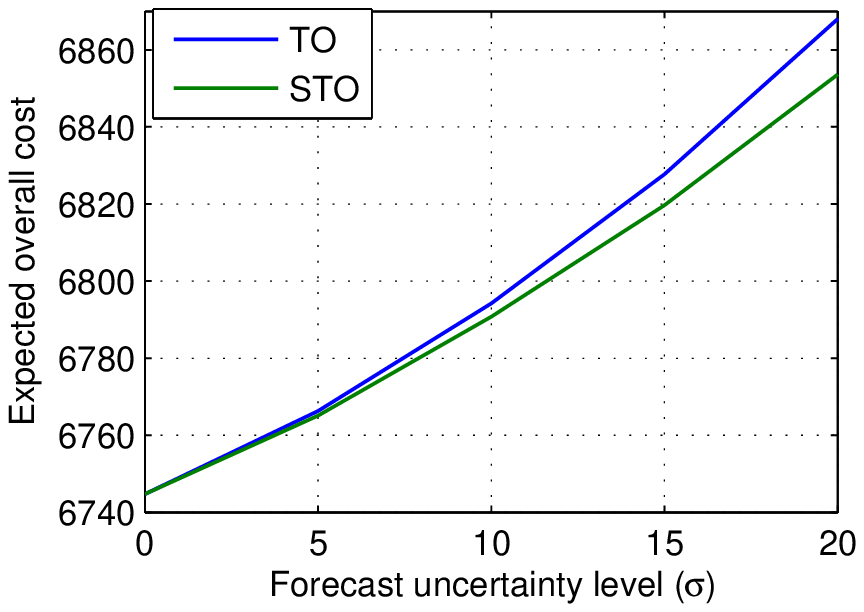}\caption{Expected overall cost.}
\end{subfigure}
\caption{Impact of the forecast uncertainty ($\sigma$).}\vspace{-1em}
\label{fig:sigma}\end{figure}

The interchange level of TO schedule does not change with $\sigma$ since it only uses the mean value $55$ of the wind production forecast. STO, on the other hand, captures the uncertainty level of the probabilistic forecast and adjusts the interchange schedule accordingly. The expected overall cost increases with the forecast uncertainty, which is observed in both TO and STO. When there is no uncertainty ($\sigma=0$), the schedules of TO and STO are the same and so do their costs.

\subsection{Example 2: a 2-region 118-bus system}
We divided the standard IEEE 118 bus system\footnote{All bus and branch indices are referred to \cite{matpower}.} into two regions: region 1 includes bus 1-12 and region 2 bus 13-118. Generator incremental cost functions, capacity limits, and load levels are the default values given in MATPOWER\cite{matpower}. We imposed the maximum capacity of $50$ MW on line 4, 6, 58 and 60. The interface transmission was not limited by default, but the impact of the interface constraint will be studied. Bus 6 and 42 were selected as the proxy buses to represent the adjacent region's network. 

To introduce randomness in the system, we assumed that three wind generators, located at bus 6, 42, and 60, produce power according to a discrete distribution. Specifically, denote the wind production by $w$, and the probabilistic forecast consists of a probability mass function $p$ and two levels of wind: $w=(10,10,10)$ and $w'=(100,200,200)$. We considered three scenarios: a high wind scenario $p=(0.1,0.9)$, a medium wind scenario $p=(0.5,0.5)$, and a low wind scenario $p=(0.9,0.1)$. TO uses the mean value $(91,181,181)$, $(55,105,105)$ and $(19,29,29)$ for each respective scenario.

\subsubsection{A baseline}\label{subsec:base118}
In this case, we verified the optimality of STO schedule with the presence of discrete randomness. All three wind scenarios were tested. Results are shown in Figure \ref{fig:118case1}-\ref{fig:118case3} and Table \ref{table:base118}.
\begin{table}\caption{Comparison of TO and STO.}\label{table:base118}\centering
\renewcommand{\arraystretch}{1.2}\begin{tabular}{|c|c|c|c|c|}\hline
Scenario&Method&$q^*$&$\mathbb{E}[\text{Cost}(q^*)]$&$\mathbb{E}[\Delta\pi(q^*)]$ \\ \hline
\multirow{2}{*}{$p=(0.1,0.9)$}&TO&$268.6$&$104751.6$&-0.25\\ \cline{2-5}
&STO&$263.1$&$104750.9$&0\\ \hline
\multirow{2}{*}{$p=(0.5,0.5)$}&TO&$242.6$&$109798.1$&-0.29\\ \cline{2-5}
&STO&$236$&$109797.1$ &0 \\ \hline
\multirow{2}{*}{$p=(0.9,0.1)$}&TO&$197$&$114806.8$&0.28\\ \cline{2-5}
&STO&$204$&$114805.8$ &0 \\ \hline
\end{tabular}
\end{table}

\begin{figure}\centering
\begin{subfigure}{0.24\textwidth}
\includegraphics[width=\textwidth]{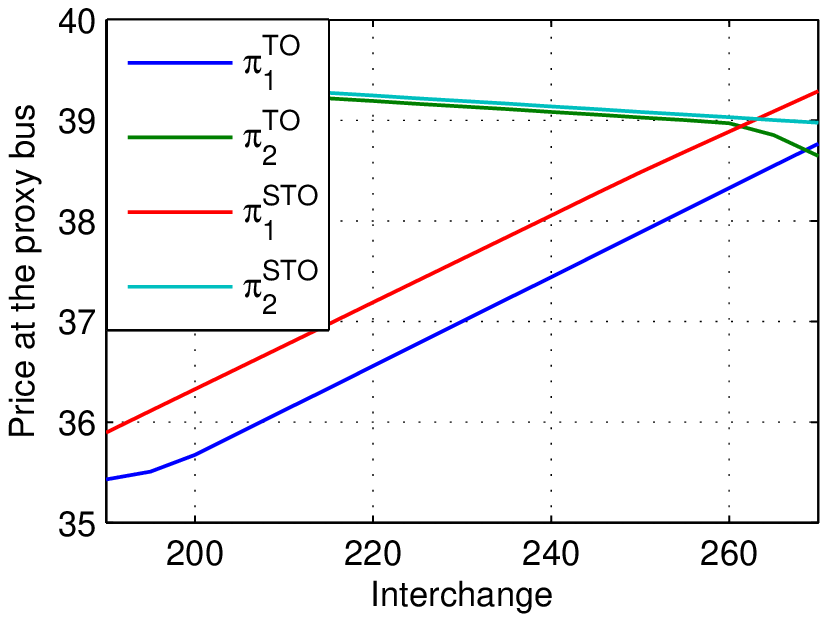}\caption{Generation supply curves.}
\end{subfigure}
\begin{subfigure}{0.24\textwidth}
\includegraphics[width=\textwidth]{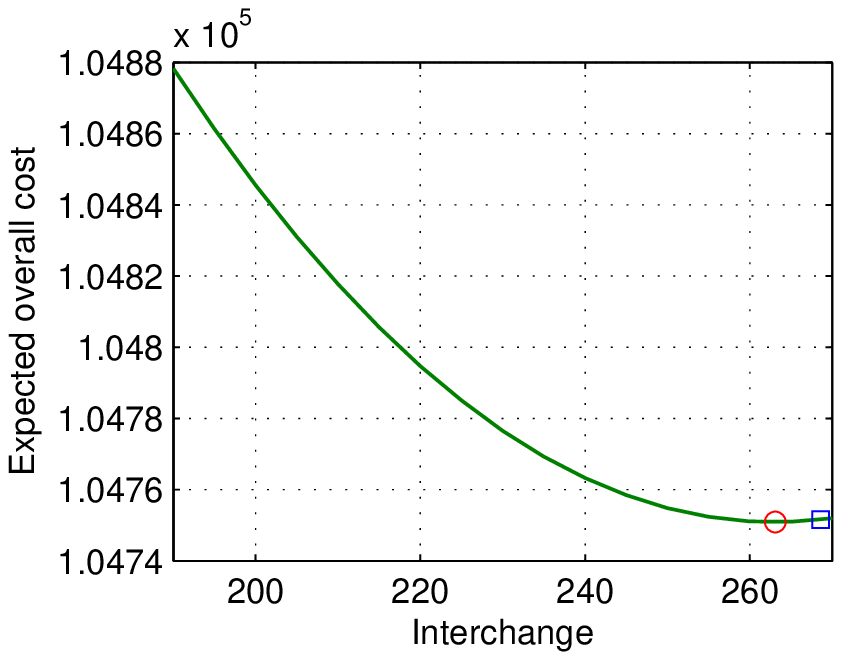}\caption{Expected overall cost.}
\end{subfigure}
\caption{High wind scenario $p=(0.1,0.9)$.}\vspace{-1em}
\label{fig:118case1}\end{figure}

\begin{figure}\centering
\begin{subfigure}{0.24\textwidth}
\includegraphics[width=\textwidth]{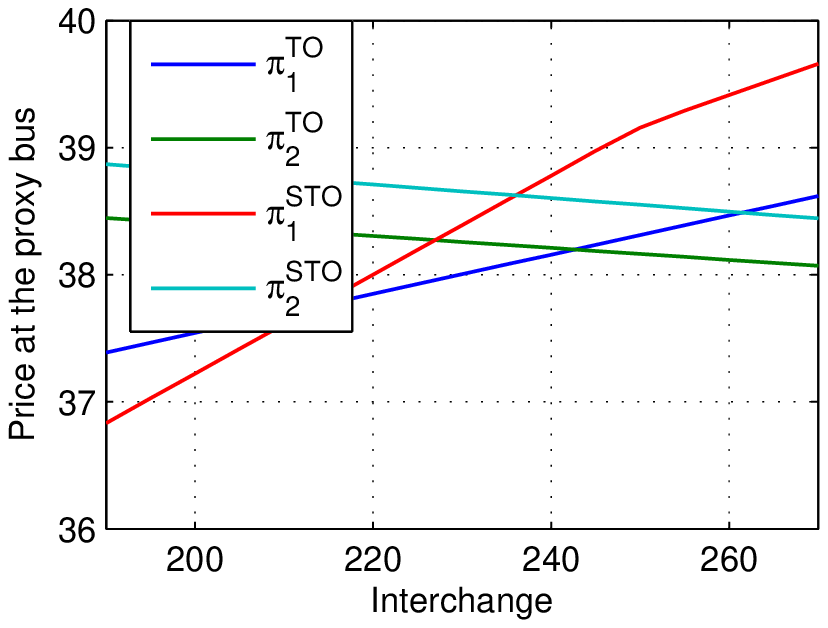}\caption{Generation supply curves.}
\end{subfigure}
\begin{subfigure}{0.24\textwidth}
\includegraphics[width=\textwidth]{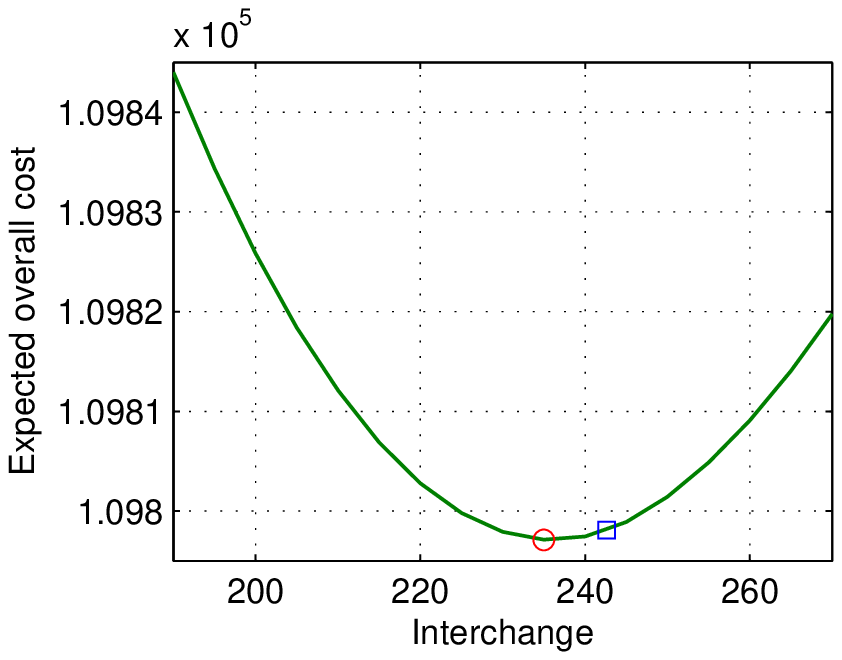}\caption{Expected overall cost.}
\end{subfigure}
\caption{Medium wind scenario $p=(0.5,0.5)$. }\vspace{-1em}
\label{fig:118case2}\end{figure}

\begin{figure}\centering
\begin{subfigure}{0.24\textwidth}
\includegraphics[width=\textwidth]{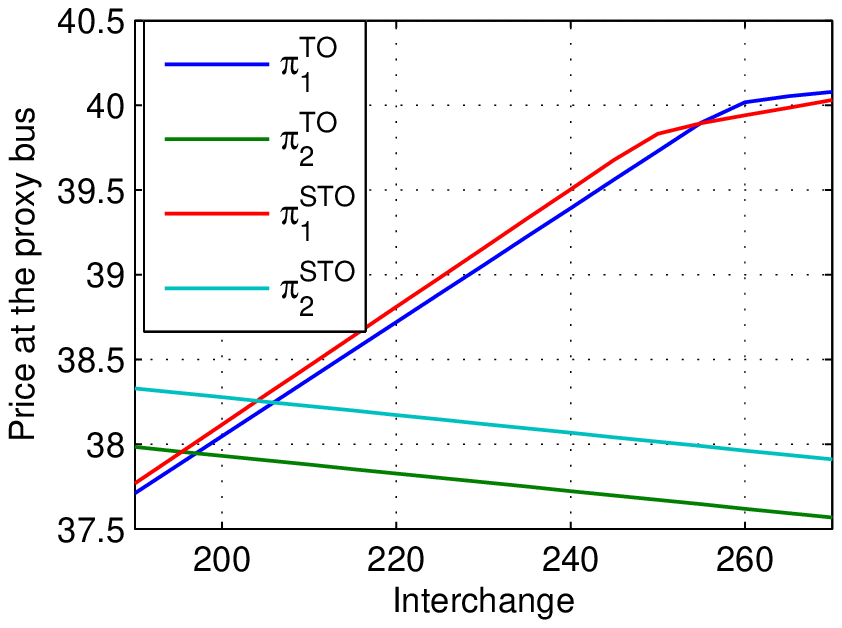}\caption{Generation supply curves.}
\end{subfigure}
\begin{subfigure}{0.24\textwidth}
\includegraphics[width=\textwidth]{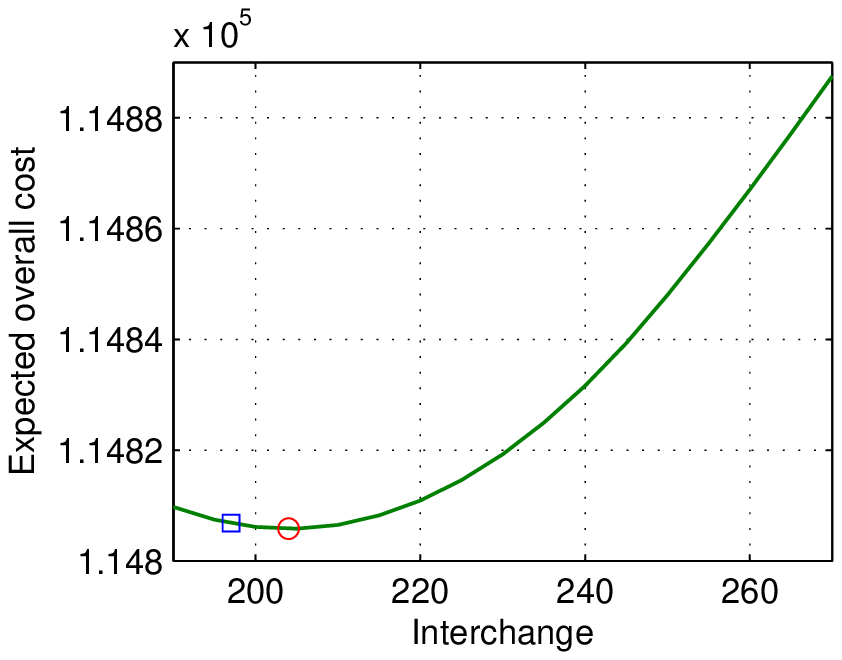}\caption{Expected overall cost.}
\end{subfigure}
\caption{Low wind scenario $p=(0.9,0.1)$. }
\label{fig:118case3}\end{figure}

Since all generation cost functions are quadratic, the price functions are continuous and piecewise affine. The performances of TO and STO schedule are similar to that in the 6-bus system example. The expected overall cost is minimized at the STO schedule in all three cases, as indicated by the red circles in Figure \ref{fig:118case1}-\ref{fig:118case3}, and the prices converge at the schedule of STO, as shown in Table \ref{table:base118}. As for TO schedules, the counter intuitive flows are observed in the high wind and medium wind scenario, and the interface under utilization happens in the low wind scenario.
\subsubsection{Impact of interface congestion}
To investigate the impact of the interface congestion, we tested all three wind scenarios again under the same setting except the interface capacity which was set as $250$ MW in this case. 

From the results shown in Table \ref{table:congestion118}, the presence of the interface constraint only influences the performances in the high wind scenario. The price separation happens in both TO and STO, because the binding interface constraint prevents the economic interface flow.
\begin{table}\caption{Impact of interface congestion.}\label{table:congestion118}\centering
\renewcommand{\arraystretch}{1.2}\begin{tabular}{|c|c|c|c|c|}\hline
Scenario&Method&$q^*$&$\mathbb{E}[\text{Cost}(q^*)]$&$\mathbb{E}[\Delta\pi(q^*)]$\\ \hline
\multirow{2}{*}{$p=(0.1,0.9)$}&TO&$250$&$104754.8$&0.6\\ \cline{2-5}
&STO&$250$&$104754.8$&0.6\\ \hline
\multirow{2}{*}{$p=(0.5,0.5)$}&TO&$242.6$&$109798.1$&-0.29\\ \cline{2-5}
&STO&$236$&$109797.1$ &0 \\ \hline
\multirow{2}{*}{$p=(0.9,0.1)$}&TO&$197$&$114806.8$&0.28\\ \cline{2-5}
&STO&$204$&$114805.8$ &0 \\ \hline
\end{tabular}\end{table}
\subsubsection{Impact of proxy bus location}
We finally tested the impact of proxy bus location in the medium wind scenario, \textit{i.e.} $p=(0.5,0.5)$. Tie line bus and internal bus (wind bus) were selected as the proxy buses. 
\begin{table}[ht]\caption{Impact of the proxy bus location.}\label{table:proxy2}\centering
\renewcommand{\arraystretch}{1.2}\begin{tabular}{|c|c|c|c|c|c|}\hline
&\multirow{2}{*}{\parbox{.8cm}{\centering Proxy Bus}}&\multicolumn{2}{|c|}{$q^*$}&\multicolumn{2}{|c|}{$\mathbb{E}[\text{Cost}(q^*)]$}\\ \cline{3-6}
&&TO&STO&TO&STO\\ \hline
\multirow{2}{*}{\parbox{.8cm}{\centering Internal Bus}}&(1,118)&244.1&238.1&109780.5&109779.7\\ \cline{2-6}
&(6,42)&242.6&236&$109798.1$&$109797.1$\\  \hline
\multirow{5}{*}{\parbox{.8cm}{\centering Tie Line Bus}}&(8,30)&253&271&109569.8&109565.3\\ \cline{2-6}
&(11,13)&78.4&126.5&110595.1&110540.4\\ \cline{2-6}
&(12,14)&28.8&-15.4&110865.5&110828.6\\ \cline{2-6}
&(12,16)&34&4.7&110847.2&110833.4\\ \cline{2-6}
&(12,117)&16.8&-23.3&110890.5&110785.7\\ \hline
\end{tabular}\end{table}
From the results presented in Table \ref{table:proxy2}, we observe that the interchange schedule and the associated expected cost are very sensitive to the location of the proxy bus in both TO and STO. With different selections of the proxy bus, the direction of the interchange schedule can be different, for example, (1,118) and (12, 117). Although there are several considerations that can guide the choice of proxy bus location \cite{Harvey03Proxy}, no theoretical results show a universal selection rule.

\section{Conclusion}\label{sec:con}
This paper presents a stochastic interchange scheduling technique that incorporates load and renewable generation uncertainties. Using the forecast of the expected LMP at the proxy bus, the proposed approach obtains the optimal interchange schedule from a deterministic optimization problem that maximizes the expected economic surplus. The essence of this technique is providing a way to reduce a two-stage stochastic optimization problem into a deterministic optimization problem with an one-dimensional decision. In addition, the proposed technique does not require any iteration between operators during the scheduling procedure. A one-time information exchange is sufficient for the optimal scheduling. 
\appendix
\begin{proof}[Proof of Lemma \ref{lem:mono}]
Denote the Lagrangian function for $(P_{2i})$ by $\mathcal{L}_{2i}, i\in\{1,2\}$. By Lemma \ref{lem:mqp}, $c_i(g_i^*(q))$ is convex and quadratic in each critical region, so the derivative exists. By the Envelope Theorem,
\[\frac{\partial c_1(g_1^*(q))}{\partial q}=\frac{\partial \mathcal{L}_{21}}{\partial q}=\pi_1(q),\frac{\partial c_2(g_2^*(q))}{\partial q}=\frac{\partial \mathcal{L}_{22}}{\partial q}=-\pi_2(q).\]

By Lemma \ref{lem:mqp},  $\pi_i(q)$ is affine in each critical region, so the derivative of $\pi_i(q)$ exists. In addition, $c_i(g_i^*(q))$ is quadratic, which implies that the second derivative of $c_i(g_i^*(q))$ (the derivative of $\pi_i(q)$) with respect to $q$ is positive. Therefore, $\pi_1(q)$ is monotonically increasing and $\pi_2(q)$ is monotonically decreasing within each critical region. Lemma \ref{lem:mqp} indicates that $\pi_i(q)$ is continuous for all $q\leq Q$, so the monotonicity of $\pi_i(q)$ is preserved for all $q\leq Q$.
\end{proof}

\Urlmuskip=0mu plus 1mu\relax
\bibliographystyle{IEEEtran}
{\bibliography{reference}}

\begin{thebibliography}{10}
\providecommand{\url}[1]{#1}
\csname url@samestyle\endcsname
\providecommand{\newblock}{\relax}
\providecommand{\bibinfo}[2]{#2}
\providecommand{\BIBentrySTDinterwordspacing}{\spaceskip=0pt\relax}
\providecommand{\BIBentryALTinterwordstretchfactor}{4}
\providecommand{\BIBentryALTinterwordspacing}{\spaceskip=\fontdimen2\font plus
\BIBentryALTinterwordstretchfactor\fontdimen3\font minus
  \fontdimen4\font\relax}
\providecommand{\BIBforeignlanguage}[2]{{%
\expandafter\ifx\csname l@#1\endcsname\relax
\typeout{** WARNING: IEEEtran.bst: No hyphenation pattern has been}%
\typeout{** loaded for the language `#1'. Using the pattern for}%
\typeout{** the default language instead.}%
\else
\language=\csname l@#1\endcsname
\fi
#2}}
\providecommand{\BIBdecl}{\relax}
\BIBdecl

\bibitem{JiTong15PESGM}
Y.~Ji and L.~Tong, ``Stochastic coordinated transaction scheduling,'' in
  \emph{Proc.~of IEEE PES General Meeting}, 2015, pp. 1--5.

\bibitem{IRIS}
\BIBentryALTinterwordspacing
{ISO New England and New York ISO}. {Inter-regional interchange scheduling
  (IRIS) analysis and options}. [Online]. Available:
  \url{http://www.iso-ne.com/pubs/whtpprs/iris_white_paper.pdf}
\BIBentrySTDinterwordspacing

\bibitem{ferc_app12}
FERC Approves Coordinated Transaction Scheduling Between New York ISO and ISO
  New England,
  \url{http://www.iso-ne.com/nwsiss/pr/2012/final_iso_ne_nyiso_cts.pdf}.

\bibitem{ferc_app14}
FERC Approves Coordinated Transaction Scheduling for PJM and NYISO,
  \url{http://www.pjm.com/~/media/about-pjm/newsroom/2014-releases/20140313-coordinated-transaction-scheduling-PJM-NYISO.ashx}.

\bibitem{ConejoEtal06Springer_Decomposition}
A.~J. Conejo, E.~Castillo, R.~Minguez, and R.~Garcia-Bertrand,
  \emph{Decomposition techniques in mathematical programming: engineering and
  science applications}.\hskip 1em plus 0.5em minus 0.4em\relax Springer
  Science \& Business Media, 2006.

\bibitem{ZhaoLitvinovZheng14TPS}
F.~Zhao, E.~Litvinov, and T.~Zheng, ``A marginal equivalent decomposition
  method and its application to multi-area optimal power flow problems,''
  \emph{IEEE Transactions on Power Systems}, vol.~1, no.~29, pp. 53--61, 2014.

\bibitem{BaldickChatterjee14COR}
R.~Baldick and D.~Chatterjee, ``Coordinated dispatch of regional transmission
  organizations: Theory and example,'' \emph{Computers \& Operations Research},
  vol.~41, pp. 319--332, 2014.

\bibitem{LiEtal15TPS}
Z.~Li, W.~Wu, M.~Shahidehpour, and B.~Zhang, ``Adaptive robust tie-line
  scheduling considering wind power uncertainty for interconnected power
  systems,'' \emph{IEEE Transactions on Power Systems}, 2015.

\bibitem{ChenThorpMount04HICSS}
J.~Chen, J.~S. Thorp, and T.~D. Mount, ``Coordinated interchange scheduling and
  opportunity cost payment: a market proposal to seams issues,'' in \emph{Proc.
  of the 37th Annual Hawaii International Conference on System Sciences}, 2004.

\bibitem{IlicLang12}
\BIBentryALTinterwordspacing
M.~Ilic and J.~Lang. {Methods of Selecting the Desired Net Interchange (DNI)
  Across Multi-Control Areas: Demonstration of Seams Solution for Large-Scale
  NPCC }. [Online]. Available:
  \url{http://www.ferc.gov/CalendarFiles/20120627090023-Wednesday_SessionA_Ilic.pdf}
\BIBentrySTDinterwordspacing

\bibitem{KimBaldick97TPS}
B.~H. Kim and R.~Baldick, ``Coarse-grained distributed optimal power flow,''
  \emph{IEEE Transactions on Power Systems}, vol.~12, no.~2, pp. 932--939,
  1997.

\bibitem{ConejoAguado98TPS_MOPF}
A.~J. Conejo and J.~A. Aguado, ``Multi-area coordinated decentralized dc
  optimal power flow,'' \emph{IEEE Transactions on Power Systems}, vol.~13,
  no.~4, pp. 1272--1278, 1998.

\bibitem{CadwaladerHarveyPopeHogan98}
M.~D. Cadwalader, S.~M. Harvey, S.~L. Pope, and W.~W. Hogan, ``Market
  coordination of transmission loading relief across multiple regions,''
  \emph{Cambridge, MA: Center for Business and Government, Harvard University},
  1998.

\bibitem{AhmadiConejoCherkaoui13TPS}
A.~Ahmadi-Khatir, A.~J. Conejo, and R.~Cherkaoui, ``Multi-area energy and
  reserve dispatch under wind uncertainty and equipment failures,'' \emph{IEEE
  Transactions on Power Systems}, vol.~28, no.~4, pp. 4373--4383, 2013.

\bibitem{Harvey03Proxy}
S.~Harvey, ``Proxy buses, seams and markets [draft],'' 2003,
  \url{http://www.hks.harvard.edu/hepg/Papers/Harvey_Proxy.Buses.Seams.Markets_5-23-03.pdf}.

\bibitem{JiTongThomas15ARXIV}
Y.~Ji, L.~Tong, and R.~J. Thomas, ``Probabilistic forecast of real-time lmp and
  network congestion,'' \emph{arXiv preprint arXiv:1503.06171}, 2015.

\bibitem{mpc_book}
\BIBentryALTinterwordspacing
F.~Borrelli, A.~Bemporad, and M.~Morari. (2014) Predictive control for linear
  and hybrid systems. [Online]. Available:
  \url{http://www.mpc.berkeley.edu/mpc-course-material}
\BIBentrySTDinterwordspacing

\bibitem{matpower}
\BIBentryALTinterwordspacing
R.~D. Zimmerman. {MATPOWER: A MATLAB power system simulation package}.
  [Online]. Available: \url{http://www.pserc.cornell.edu/matpower/}
\BIBentrySTDinterwordspacing

\end{thebibliography}
\end{document}